%% file: main.tex
\documentclass[a4paper,UKenglish,autoref]{lipics-v2019}
\usepackage[utf8]{inputenc}
\usepackage{microtype}

\usepackage{tikz}
\usetikzlibrary{positioning,shapes,petri,matrix}
\usetikzlibrary{automata}
\usetikzlibrary{arrows,backgrounds,positioning,fit,calc}

\usepackage{tikzsymbols}
\usepackage{xspace} 
\usepackage{macros} 
\usepackage{algorithm}
\usepackage{algpseudocode}

\usepackage{graphicx}
\usepackage{array,multirow}
\usepackage{url}
\usepackage{xcolor}
\usepackage{amsmath, amsfonts, amssymb}

\usepackage{thm-restate}

\Copyright{N.\ Bertrand, P.\ Bouyer, A.\ Majumdar}

\title{Synthesizing safe coalition strategies}
\titlerunning{Synthesizing safe coalition strategies}
\author{Nathalie Bertrand}{Univ. Rennes, Inria, CNRS, IRISA -- Rennes (France)}{}{http://orcid.org/0000-0002-9957-5394}{}
\author{Patricia Bouyer}{Universit\'e Paris-Saclay, ENS Paris-Saclay,
  CNRS, LSV, Gif-sur-Yvette (France)}{}{http://orcid.org/0000-0002-2823-0911}{}
\author{Anirban Majumdar}{Univ. Rennes, Inria, CNRS, IRISA -- Rennes
  (France) \and Universit\'e Paris-Saclay, ENS Paris-Saclay, CNRS, LSV, Gif-sur-Yvette (France)}{}{}{}
\authorrunning{N. Bertrand, P. Bouyer and A. Majumdar}

\nolinenumbers
\hideLIPIcs


\begin{document}
\hfuzz=2pt
\def\floatpagefraction{.9}
\def\textfraction{.1}
\def\topfraction{.9}

\maketitle

\begin{abstract}
  Concurrent games with a fixed number of agents have been thoroughly
  studied, with various solution concepts and objectives for the
  agents. In this paper, we consider concurrent games with an
  arbitrary number of agents, and study the problem of synthesizing a
  coalition strategy to achieve a global safety objective. The problem is
  non-trivial since the agents do not know \emph{a priori} how many
  they are when they start the game. We prove that the existence of a
  safe arbitrary-large coalition strategy for safety objectives is a
  \PSPACE-hard problem that can be decided in exponential space.
\end{abstract}

\section{Introduction}
\label{sec:intro}
\input{introduction}


\section{Game setting}
\label{sec:setting}
\input{setting}


\section{Solving the safety coalition problem}
\label{sec:resolution}
\input{resolution}

\subsection{Finite tree unfolding}
\label{subsec:tree-unfold}
\input{tree-unfold}

\subsection{Existence of winning coalition strategy on the tree unfolding}
\label{subsec:big-aut}
\input{big-automaton}

\subsection{\PSPACE lower bound}
\label{subsec:lower-bound}
\input{lower-bound}


\section{Future work}
\label{sec:conclusion}
\input{conclusion}


\bibliography{biblio}

\end{document}

%% file: introduction.tex
\subparagraph*{Context.}
The generalisation and everyday usage of modern distributed systems
call both for the verification and synthesis of algorithms  or strategies
running on distributed systems. Concrete
examples are cloud computing, blockchain technologies, servers with
multiple clients, wireless sensor networks, bio-chemical systems, or
fleets of drones cooperating to achieve a common
goal~\cite{CQSS-ijcai19}.  In their general form, these systems are
not only distributed, but they may also involve an arbitrary number of
agents.  This explains the interest of the model-checking community
both for the verification of parameterized
systems~\cite{Esp14,Bloem-book}, and for the synthesis of distributed
strategies~\cite{PR90}.  Our contribution is at the crossroad of those
topics.

\subparagraph*{Parameterized verification.} Parameterized
verification refers here to the verification of systems formed of an
arbitrary number of agents.
Often, the precise number of agents is unknown, yet, algorithms and
protocols running on such distributed systems are designed to operate
correctly independently of the number of agents. The automated
verification and control of crowds, \emph{i.e.},  in case the agents
are anonymous, is challenging. Remarkably, subtle changes, such as the
presence or absence of a controller in the system, can drastically
alter the complexity of the verification problems~\cite{Esp14}. In the
decidable cases, the intuition that bugs appear for a small number of
agents is sometimes confirmed theoretically by the existence of a
cutoff property, which reduces the parameterized model checking to the
verification of finitely many instances~\cite{EK00}. In the last 15
years, parameterised verification algorithms were successfully applied
to \emph{e.g.},  cache coherence protocols in uniform memory access
multiprocessors~\cite{Del03}, or the core of simple reliable broadcast
protocols in asynchronous systems~\cite{KVW15}.  When agents have
unique identifiers, most verification problems become undecidable,
especially if one can use identifiers in the code agents
execute~\cite{AK-ipl86}.

To our knowledge, there are few works on controlling parameterized
systems. Exceptions are, control strategies for (probabilistic)
broadcast networks~\cite{BFS14} and for crowds of (probabilistic)
automata~\cite{BDGGG-lmcs19,MST-corr19,CFO-fossacs20}.

\subparagraph*{Distributed synthesis.}
The problem of distributed synthesis asks whether strategies for
individual agents can be designed to achieve a global objective, in a
context where individuals have only a partial knowledge of the
environment. There are several possible formalizations for distributed
synthesis, for instance via an architecture of processes with
communication links between agents~\cite{PR90}, or using coordination
games~\cite{PR79,MW03,BKP11}.  The two settings are linked, and many
(un)decidability results have been proven, depending on various
parameters.

\subparagraph*{Concurrent games on graphs.}
  By allowing 
complex interactions between agents,
concurrent games on graphs~\cite{AlfaroHK98,AHK02} are a model of
choice in several contexts, for instance for multi-agents systems, or
for coordination or planning problems.
An arena for $n$ agents is a directed graph where the transitions are
labeled by $n$-tuples of actions (or simply words of length $n$). At
each vertex of the graph, all $n$ agents select simultaneously and
independently an action, and the next vertex is determined by the
combined move consisting of all the actions (or word formed of all the
actions). Most often, one considers infinite duration plays,
\textit{i.e.},  plays generated by iterating this process
forever. Concepts studied on multiagent concurrent games include many
borrowed from game theory, such as winning strategies (see
\emph{e.g.},~\cite{AlfaroHK98}), rationality of the agents (see
\emph{e.g.},~\cite{FKL10}), Nash equilibria (see
\emph{e.g.},~\cite{UW11a,BBMU15}).

\subparagraph*{Parameterized concurrent games on graphs.}
In a
previous work, we introduced concurrent games in which the number of
agents is arbitrary~\cite{BBM19}. These games generalize concurrent
games with a fixed number of agents, and can be seen as a succinct
representation of infinitely many games, one for each fixed number of
agents.  This is done by replacing, on edges of the arena, words
representing the choice of each of the agents by languages of finite
yet \emph{a priori} unbounded words. Such a parameterized arena can
represent infinitely many interaction situations, one for each
possible number of agents. In parameterized concurrent games, the
agents do not know \emph{a priori} the number of agents participating
to the interaction.  Each agent observes the action it plays and the
vertices the play goes through. These pieces of information may refine
the knowledge each agent has on the number of involved agents.

Such a game model raises new interesting questions, since the agents
do not know beforehand how many they are. In~\cite{BBM19}, we first
considered the question of whether Agent~$1$ can ensure a reachability
objective independently of the number of her opponents, and no matter
how they play. The problem is non trivial since Agent~$1$ must win
with a \emph{uniform} strategy. We proved that when edges are labeled
with regular languages, the problem is \PSPACE-complete; and for
positive instances one can effectively compute a winning strategy in
polynomial space.

\subparagraph*{Contribution.}
In this paper, we are interested in the coordination problem
in concurrent parameterized games, with application to distributed
synthesis. Given a game arena and an objective, the problem consists
in synthesizing for every potential agent involved in the game a
strategy that she should apply, so that, collectively, a global
objective is satisfied. In our setting, it is implicit that agents
have identifiers. However agents do not communicate; their identifier
will only be used to select the vertices the game proceeds
to. Furthermore, agents do not know how many they are, they only see
vertices which are visited, and can infer information about the number
of agents involved in the game.

\begin{figure}[h]
\centering
\begin{tikzpicture}[shorten >=1pt,node distance=12mm and 1.8cm,on grid,auto,semithick]
    \everymath{\scriptstyle}
\node[state,inner sep=1pt,minimum size=5mm, fill=green!50] (v_0) {$v_0$};
\node[state,inner sep=1pt,minimum size=5mm, fill=green!50] (v_1) [above right = of v_0] {$v_1$};
\node[state,inner sep=1pt,minimum size=5mm, fill=green!50] (v_2) [below right = of v_0] {$v_2$};
\node[state,inner sep=1pt,minimum size=5mm, fill=green!50] (v_3) [below right = of v_1] {$v_3$};
\node[state,inner sep=1pt,minimum size=5mm] (top) [above right = of v_3] {$v_4$};
\node[state,inner sep=1pt,minimum size=5mm, fill=green!50] (bot) [below right = of v_3] {$v_5$};

 \path[every node/.style={sloped,anchor=south,auto=false},-latex']
(v_0) edge 	node [above] {$(\Sigma\Sigma)^+$} (v_1)
 (v_0) edge node [below] {$\Sigma(\Sigma\Sigma)^*$} (v_2)
(v_1) edge node [above] {$\Sigma^{+}$} (v_3)
(v_2) edge node [below] {$\Sigma^{+}$} (v_3)
(v_3) edge node [above] {$(bb)^+$} node [below] {$a(aa)^*$} (top)
(v_3) edge node [above] {$(aa)^+$} node [below] {$b(bb)^*$} (bot)
;
\path[every node/.style={anchor=south,auto=false},-latex']
(bot) edge[loop right] node [right] {$\Sigma^{+}$} (bot)
(top) edge[loop right] node [right] {$\Sigma^{+}$}(top)
;
\end{tikzpicture}
\caption{Example of a parameterized arena. All unspecified transitions
lead to vertex $v_4$.}
\label{fig:running-example}
\end{figure}
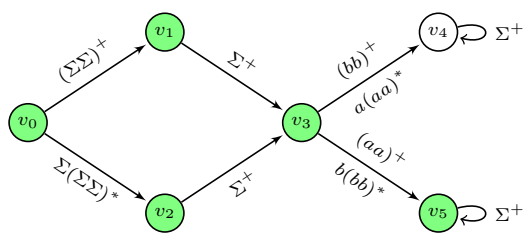

To better understand the model and the problem, consider the game
arena depicted on Fig.~\ref{fig:running-example}.  Edges are labeled
by (regular) languages.  Assuming the game starts at $v_0$, the game
proceeds as follows: a positive integer $k$ is selected by the
environment, but is not revealed to the agents; then an infinite word
$w \in \Sigma^\omega$ is selected collectively by the agents (this is
the \emph{coalition strategy}); the $n$-th letter of $w$ represents
the action played by Agent~$n$;\footnote{This is where identifiers are
implicitely used.}  depending on whether the prefix of length $k$ of
$w$ belongs to $(\Sigma \Sigma)^+$ (in case $k$ is even) or $\Sigma
(\Sigma \Sigma)^*$ (in case $k$ is odd), the game proceeds to vertex
$v_1$ or $v_2$; the process is repeated ad infinitum, generating an
infinite play in the graph. Depending on the winning condition, the
play will be winning or losing for $k$. The coalition strategy will be
said winning whenever the generated play is winning whatever the
selected number $k$ of agents is.

In this example, assuming the winning condition is to stay in the
green vertices, there is a simple winning strategy: play $a^\omega$ in
$v_0$, $v_1$ and $v_2$ (that is, all agents should play an $a$), and
if the game has gone through $v_1$ (case of an even number of agents),
then play $a^\omega$ in $v_3$ (all agents should play an $a$),
otherwise play $b^\omega$ in $v_3$ (all agents should play a
$b$). This ensures that the play never ends up in vertex $v_4$.

In this paper, we focus on safety winning conditions: the agents must
collectively ensure that only safe vertices are visited along any play
compatible with the coalition strategy in the game. We prove that the
existence of a winning coalition strategy is decidable in exponential
space, and that it is a \PSPACE-hard problem.  For positive instances,
winning coalition strategies with an exponential-size memory structure
can be synthesized in exponential space.

%% file: setting.tex
We use $\pnats$ for the set of positive natural numbers. For an
alphabet $\Sigma$ and $k \in \pnats$, $\Sigma^k$ denotes the set of
all finite words of length $k$, $\Sigma^+$ denotes the set of all
finite but non-empty words, and $\Sigma^\omega$ denotes the set of all
infinite words. For two words $u \in \Sigma^+$ and
$w \in \Sigma^+ \cup \Sigma^\omega$, we write $u \prefix w$ to denote
$u$ is a prefix of $w$, and for any $k\in \pnats$, $\wpref{w}{k}$
denotes the prefix of length $k$ of $w$ (belongs to
$\Sigma^k$).

We introduced parameterized arenas in~\cite{BBM19}, a model of arenas
with a parameterized number of agents. Parameterized arenas extend
arenas for concurrent games with a fixed number of
agents~\cite{AlfaroHK98}, by labeling the edges with languages over
finite words, which may be of different lengths. Each word represents
a joint move of the agents, for instance
$u = a_1 \cdots a_k \in \Sigma^k$ assumes there are $k$ agents, and
for every $1 \leq n \leq k$, Agent~$n$ chooses action $a_n$.

\begin{definition}
  \label{def:param-arena}
  A \emph{parameterized arena} is a tuple
  $\arena = \tuple{V, \Sigma, \Delta}$ where
  \begin{itemize}
  \item $V$ is a finite set of vertices;
  \item $\Sigma$ is a finite set of actions;
  \item $\Delta : V \times V \to 2^{\Sigma^+}$ is a partial transition
    function.
  \end{itemize}
  It is required that for every $(v, v') \in V \times V$,
  $\Delta(v,v')$ describes a regular language.
\end{definition}

Fix a parameterized arena $\arena = \tuple{V, \Sigma, \Delta}$. 
The arena $\arena$ is \emph{deterministic} if for every $v \in V$, and
every word $u \in \Sigma^+$, there is at most one vertex $v' \in V$
such that $u\in \Delta(v,v')$. The arena is assumed to be
\emph{complete}: for every $v\in V$ and $u \in \Sigma^+$, there exists
$v'\in V$ such that $u \in \Delta(v,v')$. This assumption is
natural: such an arena will be used to play games with an
arbitrary number of agents, hence for the game to be non-blocking,
successor vertices should exist whatever that number is and
irrespective of the choices of actions. 

\begin{figure}[h]
\centering
\vspace*{-.5cm}
\begin{tikzpicture}[shorten >=1pt,node distance=12mm and 3cm,on grid,auto,semithick]
    \everymath{\scriptstyle}
\node[state,inner sep=1pt,minimum size=5mm,fill=green!50] (v_0) {$v_0$};
\node[state,inner sep=1pt,minimum size=5mm,fill=green!50] (v_1) [right = of v_0] {$v_1$};
\node[state,inner sep=1pt,minimum size=5mm,fill=green!50] (v_2) [below right = 2.2cm of v_0] {$v_2$};

 \path[every node/.style={anchor=south,auto=false},-latex']
(v_0) edge [bend left] node [above] {$a^*ba^*$} (v_1)
(v_0) edge [bend right] node [below] {$a$} (v_2)
(v_2) edge [bend right] node [right] {$\Sigma^+$} (v_1)
 (v_1) edge [bend left] node [below] {$b\lor aa^+$} (v_0)
  (v_0) edge[loop above,-latex'] node [above] {$a^*ba^*$} (v_0);
\end{tikzpicture}
\caption{Example of a non-deterministic parameterized arena. Only safe
vertices (colored in green) have been depicted here. All unspecified transitions
lead to a non-safe vertex $\bot$.}
\label{fig:example-arena}
\end{figure}
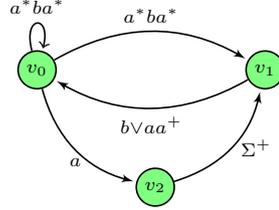

\begin{example}
  \label{ex:running-example}
  We already gave an example in the introduction. Let us give another
  example, which will be useful for illustrating the constructions
  made in the paper. Fig.~\ref{fig:example-arena} presents a
  non-deterministic parameterized arena. As such the arena is not
  complete, we assume that all unspecified moves lead to an extra
  losing vertex $\bot$, not depicted here.  If for some number of
  agents $k$ (selected by environment and not known to the agents),
  the $k$-length prefix of the word collectively chosen by the agents
  at $v_0$ belongs to $a^*ba^*$, then the play either stays at $v_0$
  or moves to $v_1$ (again selected by environment).
\end{example}

\subparagraph*{History, play and strategy.}
We fix a parameterized
arena $\arena = \tuple{V, \Sigma, \Delta}$.  A \emph{history} in
$\arena$ is a finite sequence of vertices, that is compatible with the
edges: formally, $h = v_0 v_1 \ldots v_p \in V^+$ such that for every
$1 \leq j <p$, $\Delta(v_j,v_{j{+}1})$ is defined.  We write $\Hist$
for the set of all histories.
An infinite sequence of vertices compatible with the edges is called a
\emph{play}.

A \emph{strategy} for Agent~$n$ is a mapping
$\sigma_n : \Hist \to \Sigma$ that associates an action to every
history.
A \emph{strategy profile} is a tuple of strategies, one for each
agent. Since the number of agents is not fixed \emph{a priori}, a
strategy profile is an infinite tuple of strategies:
$\stratprof = \langle \sigma_1, \sigma_2, \ldots \rangle$ = 
$(\Hist \to \Sigma)^\omega$.

\begin{table}[h]
\centering
\begin{tabular}{c||c|c|c|c|c}
\multicolumn{1}{c||}{} & $h_0$ & $h_1$& $h_2$ & $h_3$& $\ldots$\\[0.5ex] 
 \hline\hline
 $\sigma_1$ & $a$ & $b$ & $b$ & $b$ & $\ldots$\\ 
 \hline
 $\sigma_2$ & $b$ & $b$ &$b$ & $b$ & $\ldots$\\ 
 \hline
 $\sigma_3$ & $b$ & $a$ &$a$ & $a$ &$\ldots$ \\ 
 \hline
 $\vdots$&$\vdots$&$\vdots$&$\vdots$&$\vdots$
\end{tabular}
\caption{From strategy profile to coalition strategy.}
\label{tab:history}\label{link}
\end{table}

 Observe that a strategy profile can equivalently be described as a
 \emph{coalition strategy} $\stratsys : \Hist \to \Sigma^\omega$, as
 illustrated in Table~\ref{tab:history}. Indeed, if an enumeration of
 histories $(h_j)_{j\in \nats}$ is fixed, a strategy profile can be
 seen as a table with infinitely many rows --one for each agent-- and
 infinitely many columns indexed by histories. Reading the table
 vertically provides the coalition strategy view: each history is
 mapped to an $\omega$-word, obtained by concatenating the actions
 chosen by each of the agents.  Since, in this paper, we are
 interested in the existence of a winning strategy profile, it is
 equivalent to asking the existence of a winning coalition strategy
 (they may not be equivalent for some other decision problems).  In
 the sequel, we mostly take the coalition strategy view, but may
 interchangeably also consider strategy profiles.

\subparagraph*{Finite memory coalition strategies.}
Let $\stratsys : \Hist \to \Sigma^\omega$ be a coalition strategy and
$\M$ be a set. We say that the strategy $\stratsys$ \emph{uses memory
  $\M$} whenever there exist $\minit \in \M$ and applications
$\upd : \M \times V \to \M$ and
$\next : \M \times V \to \Sigma^\omega$ such that by defining
inductively $\m[h] \in \M$ by $\m[v_0] =\minit$ and
$\m[h \cdot v] = \upd(\m[h],v)$, we have that for every $h \in \Hist$,
$\stratsys(h) = \next(\m[h],\last(h))$, where $\last(h)$ is the last
vertex of history $h$. The structure $(\M,\upd)$ records information
on the history seen so far ($\m[h]$ is the memory state ``reached''
after history $h$), and $\next$ dictates how all the agents should
play.

If $\M$ is finite, then $\stratsys$ is said \emph{finite-memory}, and
if $\M$ is a singleton, then $\stratsys$ is said \emph{memoryless}
(each choice only depends on the last vertex of the history).

\subparagraph*{Realizability and outcomes.}
 For $k\in \pnats$, we say a history $h=v_0 \cdots v_p$ is
 \emph{$k$-realizable} if it corresponds to a history for $k$ agents,
 \emph{i.e.}, if for all $j < p$, there exists $u \in \Sigma^k$ with
 $u \in \Delta(v_j,v_{j{+}1})$.  A history is \emph{realizable} if it
 is $k$-realizable for some $k\in \pnats$.
Similarly to histories for finite sequences of consecutive
vertices, one can define the notions of
\emph{($k$-)realizable plays} for infinite sequences.

Given a coalition strategy $\stratsys$, an initial vertex $v_0$ and a
number of agents $k\in \pnats$, we define the $k$-\emph{outcome}
$\kout{v_0,\stratsys}$ as the set of all $k$-realizable plays 
induced by $\stratsys$ from $v_0$. Formally,
$ \kout{v_0,\stratsys} = \{v_0 v_1 \cdots \mid \ \forall j\in \pnats,\
\wpref{\stratsys(v_0 \cdots v_j)}{k} \in \Delta(v_j,v_{j{+}1}) \}$.
Note that the completeness assumption ensures that the set
$\kout{v,\stratsys}$ is not empty. Then the \emph{outcome} of
coalition strategy $\stratsys$ is simply
$\out{v_0,\stratsys}=\bigcup_{k\in \pnats} \kout{v_0,\stratsys}$.

\subparagraph*{The safety coalition problem.}
We are now in a position to define our problem of interest. Given an
arena $\mathcal{A} = \tuple{V, \Sigma, \Delta}$, a set of \emph{safe}
vertices $S \subseteq V$ defines a \emph{parameterized safety game}
$\mathcal{G} = (\mathcal{A},S)$. Without loss of generality we assume
from now that $V \setminus S$ are sinks.
A coalition strategy $\stratsys$ from
$v_0$ in the safety game $\mathcal{G} = (\mathcal{A},S)$ is said
\emph{winning} if all induced plays only visit vertices from $S$:
$\out{v,\stratsys} \subseteq S^\omega$. Our goal is to study the
decidability and complexity of the existence of winning coalition
strategies, and to synthesize such winning coalition strategies when
they exist. We therefore introduce the following decision problem:
\medskip
   
\noindent\fbox{\begin{minipage}{.99\linewidth}
    \textsc{Safety coalition problem} \\
    {\bf Input}: A parameterized safety game $\mathcal{G} =
    (\mathcal{A},S)$ and an initial vertex $v_0$.\\
    {\bf Question}: Does there exist a coalition strategy $\stratsys$
    such that $\out{v_0,\stratsys} \subseteq S^\omega$?
   \end{minipage}}
\smallskip

The safety coalition problem is a coordination problem: agents should
agree on a joint strategy which, when played in the graph and no
matter how many agents are involved, the resulting play is safe. Note
that, due to the link between coalition strategies and tuples of
individual strategies mentioned on Page~\pageref{link}, the coalition
strategies are distributed: the only information required for an agent
to play her strategy is the history so far, not the number of agents
selected by the environment; however she can infer some information
about the number of agents from the history; this is for instance the
case at vertex $v_3$ in the example of Fig.~\ref{fig:example-arena}.
Note that there is no direct communication between agents.

   \begin{example}
     \label{ex:1}
     We have already given in the introduction a winning coalition
     strategy for the game in Fig.~\ref{fig:running-example}.
     On the arena in Fig.~\ref{fig:example-arena},
     assuming $\bot$ is the only unsafe vertex, one can also show
     that the agents have a winning coalition strategy $\stratsys$
     from $v_0$ to stay within green (\emph{i.e.}, safe) vertices.
     Consider the coalition strategy $\stratsys$ such that
     $\stratsys(v_0) = aba^\omega$, $\stratsys(v_0 v_2) = a^\omega$,
     $\stratsys(v_0 v_1) = a^\omega$, and
     $\stratsys(v_0 v_2 v_1) = b^\omega$.
     Intuitively, on playing $aba^\omega$ from $v_0$,
     in one step, the game either
     stays in $v_0$ (which is `safe') or moves to $v_2$ (in case the
     number of agents $k = 1$) or to $v_1$ (in case $k \ge 2$); from
     $v_1$, depending on history, coalition plays either $b^\omega$
     (when the history is $v_0 v_2 v_1$ and hence $k = 1$) or
     $a^\omega$ (otherwise) which leads the game back to $v_0$ (note
     that at vertex $v_2$, choice of actions of the agents is not
     important, they can collectively play any $\omega$-word).
     However, one can show that there is no memoryless coalition
     winning strategy.  Indeed, the coalition strategy $a^\omega$ from
     $v_1$ is losing for $k = 1$, similarly $b^\omega$ from $v_1$ is
     losing for $k \ge 2$, and any other strategy is also losing. For
     instance, $ba^\omega$ from $v_0$ is losing because if the game
     moves to $v_1$, coalition has no information on the number of
     agents and hence any word from $v_1$ will be losing ($a^\omega$
     is losing for $k=1$, $b^\omega$ is losing for $k \ge 2$, and
     similarly for other words).
   \end{example}

The rest of the paper is devoted to the proof of the following
theorem:

\begin{theorem}
\label{theo:main}
  The safety coalition problem can be solved in exponential space, and
  is \PSPACE-hard. 
  For positive instances, one can synthesize a winning coalition
  strategy in exponential space which uses exponential memory; the
  exponential blowup in the size of the memory is tight.
\end{theorem}

%% file: resolution.tex
This section is devoted to the proof of Theorem~\ref{theo:main}.
To prove the decidability and establish the complexity upper bound, we
construct a tree unfolding of the arena, which is equivalent for
deciding the existence of a winning coalition strategy. The unfolding
is finite because, if a vertex is repeated along a play, the coalition
can play the same $\omega$-word as in the first visit, which will be
formalized in Section~\ref{subsec:tree-unfold}. We can then show how
to solve the safety coalition problem at the tree level in
Section~\ref{subsec:big-aut}. Synthesis and memory usage are analyzed
in Section~\ref{subsec:synthesis}, and the running example game is
discussed in Section~\ref{subsec:illustration}

The hardness result is shown in
Section~\ref{subsec:lower-bound} by a reduction from the
\textsf{QBF-SAT} problem which is known to be
\PSPACE-complete~\cite{StockmeyerM73}.

%% file: tree-unfold.tex
From a parameterized safety game $\game = (\arena,S)$, we construct a
finite tree as follows: we unfold the arena $\arena$ until either some
vertex is repeated along a branch or an unsafe vertex is reached.  The
nodes of the tree are labeled with the corresponding vertices and the
edges are labeled with the same regular languages as in the arena
$\arena$. The intuition behind this construction is that if a vertex
is repeated in a winning play in $\arena$, since the winning condition
is a safety one, the coalition can play the same strategy as it played
in the first occurrence of the vertex. Note however that multiple
nodes in the tree may have the same label but different (winning)
strategies depending on the history (recall Example~\ref{ex:1}). This
is the reason why we need to consider a tree unfolding abstraction and
not a DAG abstraction.

We assume the concept of tree is known.  Traditionally, we call a node
$\node'$ a \emph{child} of $\node$ (and $\node$ the \emph{parent} of
$\node'$) if $\node'$ is an immediate successor of $\node$ according
to the edge relation; and $\node$ an \emph{ancestor} of $\node'$ if
there exists a path from $\node$ to $\node'$ in the tree.
  \begin{definition}
    Let $\game = (\arena=\tuple{V, \Sigma, \Delta},S)$ be a
    parameterized safety game and $v_0 \in V$ an initial vertex.
    The \emph{tree unfolding} of $\game$ is the tree
    $\tree = \langle {N,E, \ell_N,\ell_E} \rangle$ rooted at
    $\node_0 \in N$, where $N$ is the finite set of nodes,
    $E \subseteq N \times N$ is the set of edges, $\ell_N : N \to V$
    is the node labeling function, $\ell_E: N \times N \to 2^{\Sigma^+}$
    is the edge labeling function, and:
\begin{itemize}
\item
  the root $\node_0$ satisfies $\ell_N(\node_0) = v_0$;
\item $\forall \node \in N$, if $\ell_N(\node)\in S$ and for
  every ancestor $\node''$ of $\node$,
  $\ell_N(\node'') \ne \ell_N(\node)$, then $\forall v' \in V$ such
  that $\Delta(v,v')$ is defined, there is $\node'$ a child of $\node$ with
  $\ell_N(\node') = v'$ and $\ell_E(\node,\node') = \Delta(v,v')$;
   otherwise, the node $\node$ has no successor.
\end{itemize}
\end{definition}

Each node in $\tree$ corresponds to a unique history in $\game$, and
the unfolding is stopped when a vertex repeats or an unsafe vertex is
encountered.  The set of nodes can be partitioned into
$N = \nodeint \sqcup \nodeleaf$ where $\nodeint$ is the set of
internal nodes and $\nodeleaf$ are the leaves of $\tree$ (some leaves
are unsafe, some leaves have an equilabeled ancestor). By
construction, the height of $\tree$ is bounded by $|V| +1$ and its
branching degree is at most $|V|$. The tree unfolding of $\game$ is
hence at most in $O(|V|^{|V|})$ 
(and the exponential blowup is unavoidable in general). 

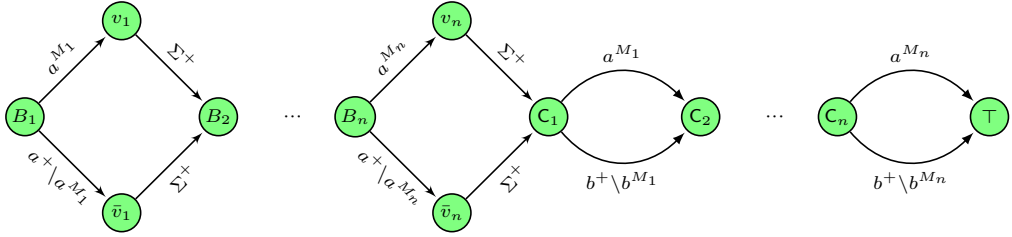
\begin{figure}[h]
\begin{center}
\begin{tikzpicture}[shorten >=1pt,node distance=1.8cm,on grid,auto,semithick]
    \everymath{\scriptstyle}
\node[state,inner sep=1pt,minimum size=5mm, fill=green!50] (B_1) {$B_1$};
\node[state,inner sep=1pt,minimum size=5mm, fill=green!50] (v_1) [above right = of B_1] {$v_1$};
\node[state,inner sep=1pt,minimum size=5mm, fill=green!50] (v_2) [below right = of B_1] {$\bar{v}_1$};
\node[state,inner sep=1pt,minimum size=5mm, fill=green!50] (B_2) [below right = of v_1] {$B_2$};
\node (dots) [right = 1cm of B_2] {$\cdots$};
\node[state,inner sep=1pt,minimum size=5mm, fill=green!50] (B_n) [ right = of B_2] {$B_n$};
\node[state,inner sep=1pt,minimum size=5mm,fill=green!50] (v_{2n-1}) [above right =  of B_n] {$v_n$};
\node[state,inner sep=1pt,minimum size=5mm, fill=green!50] (v_{2n}) [below right =  of B_n] {$\bar{v}_n$};
\node[state,inner sep=1pt,minimum size=5mm, fill=green!50] (C_1) [below right = of v_{2n-1}] {$\vclause{1}$};

\node[state,inner sep=1pt,minimum size=5mm, fill=green!50] (C_2) [ right = 2cm of C_1] {$\vclause{2}$};
\node (dots) [right = 1cm of C_2] {$\cdots$};
\node[state,inner sep=1pt,minimum size=5mm, fill=green!50] (C_{n-1}) [ right = of C_2] {$\vclause{n}$};
\node[state,inner sep=1pt,minimum size=5mm, fill=green!50] (C_n) [ right = 2cm of C_{n-1}] {$\top$};
 \path[every node/.style={sloped,anchor=south,auto=false},-latex']
(B_1) edge 	node [above] {$a^{M_1}$} (v_1)
 (B_1) edge node [below] {$a^+ \setminus a^{M_1}$} (v_2)
(v_1) edge node [above] {$\Sigma^{+}$} (B_2)
(v_2) edge node [below] {$\Sigma^{+}$} (B_2)
(B_n) edge node [above] {$a^{M_n}$} (v_{2n-1})
(B_n) edge node [below] {$a^+ \setminus a^{M_n}$} (v_{2n})
(v_{2n-1}) edge  node [above] {$\Sigma^{+}$} (C_1)
(v_{2n}) edge  node [below] {$\Sigma^{+}$} (C_1)
;
 \draw [-latex] (C_1) .. controls +(50:1cm) and +(135:1cm) .. (C_2) node [midway,above] {$a^{M_1}$}; 
 \draw [-latex] (C_1) .. controls +(-50:1cm) and +(-135:1cm) .. (C_2) node [midway,below] {$b^+ \setminus b^{M_1}$}; 
  \draw [-latex] (C_{n-1}) .. controls +(50:1cm) and +(135:1cm) .. (C_n) node [midway,above] {$a^{M_n}$}; 
 \draw [-latex] (C_{n-1}) .. controls +(-50:1cm) and +(-135:1cm) .. (C_n) node [midway,below] {$b^+ \setminus b^{M_n}$}; 
\end{tikzpicture}
\caption{Example arena such that the tree unfolding is exponential.
  All unspecified transitions lead to the sink losing vertex
  $\bot$. Set $M_i$ denotes multiples of the $i$-th prime number. For
  any play reaching $\vclause{1}$, for every $i$, the number of agents
  is in $M_i$ iff the play went through $v_i$.}
\label{fig:worstcase-example}
\end{center}
\end{figure}
 The
exponential bound is reached by a a family $(\arena_n)_{n \in \pnats}$
of deterministic arenas, shown in Fig.~\ref{fig:worstcase-example},
which is an extension of the example in
Fig.~\ref{fig:running-example}, with $2n$ many blocks (and, $O(n)$
many vertices).  Observe that to win the game, coalition needs to keep
track of the full histories in the first $n$ blocks, and there are
exponentially many such histories; moreover, each such history
corresponds to a different node in its unfolding
tree.

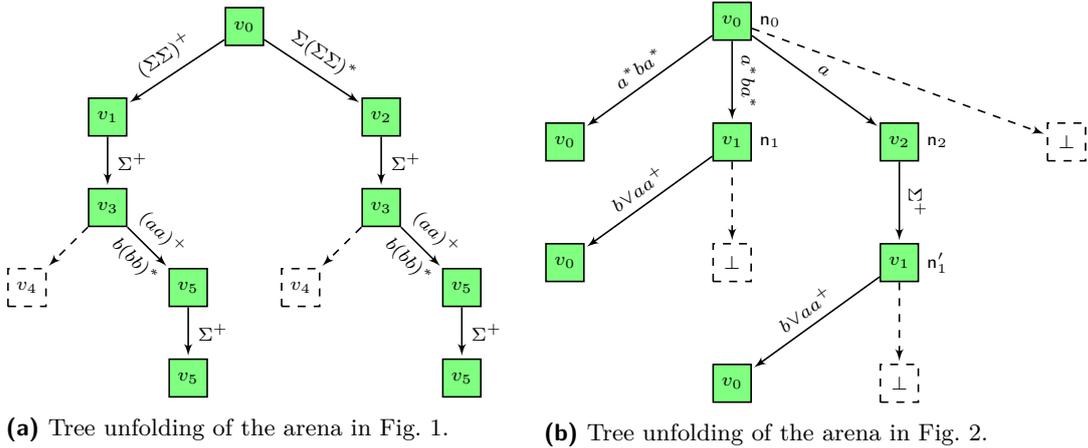
\begin{figure}[h]
\begin{subfigure}[l]{0.5\textwidth}
\begin{tikzpicture}[shorten >=1pt,node distance=12mm and 1.8cm,on grid,auto,semithick]
    \everymath{\scriptstyle}
\node[rectangle,draw,inner sep=1pt,minimum size=5mm, fill=green!50] (v_0) {$v_0$};
\node[rectangle,draw,inner sep=1pt,minimum size=5mm, fill=green!50] (v_1) [below left = of v_0] {$v_1$};
\node[rectangle,draw,inner sep=1pt,minimum size=5mm, fill=green!50] (v_2) [below right = of v_0] {$v_2$};
\node[rectangle,draw,inner sep=1pt,minimum size=5mm, fill=green!50] (v_31) [below  = of v_1] {$v_3$};
\node[rectangle,draw,inner sep=1pt,minimum size=5mm, fill=green!50] (v_32) [below  = of v_2] {$v_3$};
\node[rectangle,draw,inner sep=1pt,minimum size=5mm] (v_41) [below left = 1.5cm of v_31,dashed] {$v_4$};
\node[rectangle,draw,inner sep=1pt,minimum size=5mm, fill=green!50] (v_51) [below right = 1.5cm of v_31] {$v_5$};
\node[rectangle,draw,inner sep=1pt,minimum size=5mm, fill=green!50] (v_52) [below  = of v_51] {$v_5$};
\node[rectangle,draw,inner sep=1pt,minimum size=5mm] (v_411) [below left = 1.5cm of v_32,dashed] {$v_4$};
\node[rectangle,draw,inner sep=1pt,minimum size=5mm, fill=green!50] (v_511) [below right = 1.5cm of v_32] {$v_5$};
\node[rectangle,draw,inner sep=1pt,minimum size=5mm, fill=green!50] (v_521) [below  = of v_511] {$v_5$};

 \path[every node/.style={,anchor=south,auto=false},-latex']
(v_0) edge 	node [above,sloped] {$(\Sigma\Sigma)^+$} (v_1)
 (v_0) edge node [above,sloped] {$\Sigma(\Sigma\Sigma)^*$} (v_2)
(v_1) edge node [right] {$\Sigma^{+}$} (v_31)
(v_2) edge node [right] {$\Sigma^{+}$} (v_32)
(v_31) edge [dashed] node [above,sloped] {$$} node [below,sloped] {$$} (v_41)
(v_31) edge node [above,sloped] {$(aa)^+$} node [below,sloped] {$b(bb)^*$} (v_51)
(v_32) edge [dashed] node [above,sloped] {$$} node [below,sloped] {$$} (v_411)
(v_32) edge node [above,sloped] {$(aa)^+$} node [below,sloped] {$b(bb)^*$} (v_511)
(v_51) edge node [right] {$\Sigma^{+}$} (v_52)
(v_511) edge node [right] {$\Sigma^{+}$} (v_521)
;
\end{tikzpicture}
\caption{Tree unfolding of the arena in
  Fig.~\ref{fig:running-example}.}
\label{fig:unfold1}
\end{subfigure}
\begin{subfigure}[l]{0.5\textwidth}
\begin{tikzpicture}[shorten >=1pt,node distance=1.6cm and 2.2cm,on grid,auto,semithick]
    \everymath{\scriptstyle}
\node[rectangle,draw,inner sep=1pt,minimum size=5mm, fill=green!50]
(v_0) {$v_0$};
\node [right =.5cm of v_0] {$\node_0$};
\node[rectangle,draw,inner sep=1pt,minimum size=5mm, fill=green!50]
(v_01) [below left = of v_0] {$v_0$};
\node[rectangle,draw,inner sep=1pt,minimum size=5mm, fill=green!50]
(v_1) [below  = of v_0] {$v_1$}; 
\node [right =.5cm of v_1] {$\node_1$};
\node[rectangle,draw,inner sep=1pt,minimum size=5mm, fill=green!50]
(v_2) [below right = of v_0] {$v_2$}; 
\node [right =.5cm of v_2] {$\node_2$};
\node[rectangle,draw,inner sep=1pt,minimum size=5mm]
(bot1) [right = of v_2,dashed] {$\bot$};

\node[rectangle,draw,inner sep=1pt,minimum size=5mm, fill=green!50]
(v_02) [below  left = of v_1] {$v_0$}; 
\node[rectangle,draw,inner sep=1pt,minimum size=5mm]
(bot2) [right = of v_02,dashed] {$\bot$};

\node[rectangle,draw,inner sep=1pt,minimum size=5mm, fill=green!50]
(v_12) [below   =  of v_2] {$v_1$}; 
\node [right =.5cm of v_12] {$\node'_1$};

\node[rectangle,draw,inner sep=1pt,minimum size=5mm, fill=green!50]
(v_03) [below   left =  of v_12] {$v_0$};
\node[rectangle,draw,inner sep=1pt,minimum size=5mm]
(bot3) [right = of v_03,dashed] {$\bot$};

 \path[every node/.style={anchor=south,auto=false},-latex']
(v_0) edge node [above,sloped] {$a^*ba^*$} (v_01) 
(v_0) edge node [above,sloped] {$a^*ba^*$} (v_1) 
(v_0) edge [dashed] node [above,sloped] {} (bot1)
(v_0) edge node [above,sloped] {$ a$} (v_2) 
(v_2) edge node [above,sloped] {$\Sigma^+ $} (v_12) 
(v_1) edge node [above,sloped] {$b\lor aa^+$} (v_02) 
(v_1) edge [dashed] node [above,sloped] {} (bot2)
(v_12) edge node [above,sloped] {$b\lor aa^+$} (v_03)
(v_12) edge [dashed] node [above,sloped] {} (bot3) ;
\end{tikzpicture}
\caption{Tree unfolding of the arena in Fig.~\ref{fig:example-arena}.}
\label{fig:unfold2}
\end{subfigure}
\caption{Tree unfolding examples (green nodes correspond to safe vertices).
Notice here that the unsafe leaves (and the edges leading to them) are presented with
dashed rectangles (resp. arrows).}
\label{fig:unfolding}
\end{figure}

\begin{example}
  Fig.~\ref{fig:unfold1} and \ref{fig:unfold2} represent the tree
  unfoldings of the parametrized arenas depicted in
  Fig.~\ref{fig:running-example} and~\ref{fig:example-arena},
  respectively. On the left picture, the node names are avoided, and
  in all cases their labels are written within the nodes.
  The leaf nodes that correspond to unsafe vertices (and the edges
  leading to them) are presented with dashed rectangles
  (respectively,~arrows).  Notice that any leaf node is either labeled
  with an unsafe vertex (for instance, $v_4$ in
  Fig.~\ref{fig:unfold1}) or it has a unique ancestor with the same
  label.  These two criteria ensure the tree is always finite (along
  all branches, some vertex has to repeat within $|V|$ many steps).
  However, multiple internal nodes in different branches can have same
  label but,  coalition might have different
  (winning) strategies depending on their respective histories.
  \end{example}

Let $\game = (\arena=\tuple{V, \Sigma, \Delta},S)$ be a parameterized
safety game with an initial vertex $v_0$ and
$\tree = \langle {N,E, \ell_N,\ell_E} \rangle$ be the tree unfolding
corresponding to $\arena$ with root $\node_0$.  We define the
coalition game on $\tree$ as follows.

\subparagraph*{History, play and strategy.}
Histories in $\tree$ are defined similarly as in $\game$ (except
vertices are replaced by nodes); the set of such histories is denoted
$\Histree$.
 A \emph{history} in $\tree$ is a finite sequence of nodes
 $H = \node_0 \node_1 \ldots \node_p \in N^+$ such that for every
 $0 \le j < p$, $(\node_{j},\node_{j+1}) \in E$.  We denote by
 $\Histree$ the set of all histories in $\tree$.
A \emph{play} in $\tree$ is a maximal history, \emph{i.e.}, a finite
sequence of nodes ending with a leaf, thus in
$\nodeint^+\cdot\nodeleaf$.
Note that, contrary to the definition of a play
in $\arena$, a play in $\tree$ is a \emph{finite} sequence of nodes
ending in a leaf.

A \emph{coalition strategy} in the unfolding tree is a mapping
$\stratree : \nodeint \to \Sigma^\omega$ that assigns to every
internal node $\node \in \nodeint$ an $\omega$-word
$\stratree(\node)$.  Notice that a coalition strategy in $\tree$ is by
definition memoryless (on $N$) which, as we will see later, is
sufficient to capture winning strategies of the coalition in $\game$.
We furthermore extend the definition of node labeling function
$\ell_N$ to a history (resp.~play) in the usual way.

Similarly to the parameterized arena setting, we define in a natural
way the notions of $k$-realizability and of realizability for
histories and plays.
We also define for a coalition strategy $\stratree$ in $\tree$ (rooted
at $\node_0$), and $k\in \pnats$ the sets
$\koutree{\node_0,\stratree}$ and $\outree{\node_0,\stratree}$.

A coalition strategy $\stratree$ in $\tree$ from $\node_0$ is
\emph{winning} for the safety condition defined by the safe set $S$ if
every play in $\outree{\node_0,\stratree}$ ends in a leaf with label
in $S$, \emph{i.e.}, if for every
$R = \node_0 \ldots \node_p \in \outree{\node_0, \stratree}$,
$\ell_N(\node_p) \in S$, written
$\ell_N(\out{\node_0,\stratree}) \subseteq S^+$ for short.

\subparagraph*{Correctness of the tree unfolding.}
Next we show the equivalence of the winning strategies in the safety
coalition game, and in the corresponding tree unfolding:
   \begin{restatable}{lemma}{strateqv}
\label{lem:str-eqv}
Let $\game= (\arena=\tuple{V, \Sigma, \Delta},S)$ be a parameterized
safety game and $v_0 \in V$ and
$\tree = \langle {N,E, \ell_N,\ell_E} \rangle$ be the associated tree
unfolding with root $\node_0$. There exists a winning coalition
strategy from $v_0$ in $\game$ iff
there exists a winning coalition strategy from $\node_0$ in $\tree$.
\end{restatable}

\begin{proof}
  Assume first that the coalition of agents has a winning strategy
  $\stratsys$ in $\game$. Any history $H \in \Histree$ can be
  projected to the history $\ell_N(H) \in \Hist$. We can hence define
  for every $\node \in \nodeint$,
  $\stratree(\node) = \stratsys(\ell_N(\iota(\node)))$, where $\iota$
  is the bijection mapping nodes to histories in $\tree$. 
 To prove that $\stratree$ is winning in $\tree$, consider any play
  $R = \node_0 \cdots \node_p$ in $\outree{\node_0,\stratree}$ and let
  $\rho = \ell_N(R) = v_0 \cdots v_p$ be its projection in $\game$.
  By construction $\ell_E(\node_i,\node_{i+1}) = \Delta(v_i,v_{i+1})$
  for each $i < p$, and hence from the definition of $\stratree$,
  $\rho$ is a history in $\game$ induced by $\stratsys$.  Since
  $\stratsys$ is winning, $\rho$ only visits \emph{safe} vertices.  In
  particular, $\ell_N(\node_p) \in S$.  Since this is true for every
  play induced by $\stratree$, strategy $\stratree$ is winning from
  $\node_0$ in $\tree$.

  \smallskip For the other direction, assume that $\stratree$ is a
  winning coalition strategy from $\node_0$ in $\tree$. The tree will
  be the basis of a memory structure sufficient to win the game; we
  thus explain how histories in $\game$ can be mapped to nodes of
  $\tree$.  We first define a mapping $\zip : \Hist \to \Hist$ that
  summarizes any history in $\arena$ to its \emph{virtual history}
  where each vertex appears at most once. Intuitively, $\zip$ greedily
  shortens a history by appropriately removing the loops until an
  unsafe vertex is encountered (if any). The mapping $\zip$ is defined
  inductively, starting with $\zip(v_0) = v_0$, and letting for every
  $h\in \Hist$ and every $v'\in V$ such that $h \cdot v' \in \Hist$,
   \[
  \zip(h \cdot v') = \begin{cases}
    \zip(h) \cdot v' & \textrm{if } 
     v' \textrm{ does not appear in } \zip(h)\\
    v_0 \ldots v' \prefix \zip(h) & \textrm{otherwise}\\
  \end{cases}
\]
The mapping $\zip$ is well-defined (by construction, for every history
$h$, any vertex appears at most once in $\zip(h)$, so that when $v'$
appears in $\zip(h)$, there is a unique prefix of $\zip(h)$ ending
with $v'$). Note that, since unsafe vertices are sinks, as soon as $h$
reaches an unsafe vertex, the value of $\zip(h)$ stays unchanged.

\begin{restatable}{lemma}{betabijection}
\label{lem:betabijection}
  The application $\beta \colon \node \mapsto \ell_N(\iota(\node))$
  defines a bijection between
  $\nodeint \cup \{\node \in \nodeleaf \mid \ell_N(\node) \notin S\}$
  and the set $Z = \{\zip(h) \mid h \in \Hist\}$.
\end{restatable}

\begin{proof}
 It is first obvious that this application
 is injective, since two nodes of the tree corresponds to different
 histories in $\arena$ which all belong to $Z$.
 
 This application is surjective: pick $h \in Z$; then, $h$ has no
 repetition; furthermore it forms a real history in $\game$, which
 implies that it can be read as the label of some history in the tree
 unfolding.
\end{proof}

We write $\alpha = \beta^{-1}$.
Using the $\zip$ function and $\alpha$, from a coalition strategy
$\stratree$ in $\tree$, we define a
coalition strategy $\stratsys$ in $\game$ by applying $\stratsys$ to
the virtual histories:
for every history $h = v_0 \ldots v_p$ in $\game$ we let
$\stratsys(h) = \stratree(\alpha(\zip(h))$ whenever
$\alpha(\zip(h)) \in \nodeint$ and $\stratsys(h)$ is set arbitrarily
otherwise (recall that if $\alpha(\zip(h))$ is a leaf node, then $h$ is
actually already a losing history). 
\

Towards a contradiction, assume that $\stratsys$ is not winning in
$\game$.  Consider, some number of agents $k\in \pnats$, and a losing
play with $k$~agents:
$\rho = v_0 v_1 \ldots \in \kout{v_0,\stratsys}$.
Let $h' = v_0 v_1 \ldots v_q \prefix \rho$ be the shortest prefix of
$\rho$ ending in an unsafe vertex $v_q \notin S$, and write
$\zip(h') = v_0 v_{i_1} \ldots v_{q}$ for the corresponding virtual
history. 
By definition of $\stratsys$,
$\zip(h')$ is a $k$-outcome of $\stratsys$ from $v_0$. Moreover, the
corresponding play
$R = \iota(\alpha(\zip(h'))) = \node_0 \node_{i_1}\ldots \node_{q}$ in
$\tree$,
  belongs to the $k$-outcome
  of $\stratree$ from $\node_0$. Since $v_{q} \notin S$, $\stratree$
is not winning in $\tree$; which is a contradiction.  We conclude that
$\stratsys$ is a winning coalition strategy in $\game$.
\end{proof}

\begin{example}
  We illustrate the $\zip$ function on the arena in
  Fig.~\ref{fig:example-arena}.  Take $h = v_0 v_1 v_0 v_1$.  First,
  $\zip(v_0) = v_0$; then
  $\zip(v_0 v_1) = \zip(v_0)\cdot v_1 = v_0 v_1$;
  $\zip(v_0 v_1 v_0) = v_0$ (which is the unique prefix of
  $\zip(v_0 v_1) = v_0 v_1$, ending at $v_0$); finally
  $\zip(v_0 v_1 v_0 v_1) = \zip(v_0 v_1 v_0) \cdot v_1 = v_0 v_1$.
  Then the function $\alpha$ uniquely maps each virtual history
  (\emph{i.e.}, $\zip(h)$) ending at a safe vertex to an internal node
  in the tree, which is the heart of the proof of
  Lemma~\ref{lem:str-eqv}.
\end{example}

%% file: big-automaton.tex
In the previous subsection, we showed that the safety coalition
problem reduces to solving the existence of a
winning coalition strategy in the associated finite tree unfolding.

To solve the latter, from the tree unfolding $\tree$, we construct a
deterministic (safety) automaton over the alphabet $\Sigma^{m}$, where
$m=|\nodeint|$, which accepts the $\omega$-words corresponding to
winning coalition strategies in $\tree$.  More precisely, since
$(\Sigma^m)^\omega$ and $(\Sigma^\omega)^m$, understood as the set of
$m$-tuples of $\omega$-words over $\Sigma$, are in one-to-one
correspondence, an infinite word $\word \in (\Sigma^m)^\omega$
corresponds to $m$ infinite words $\word_\node$, one for each
internal node $\node \in \nodeint$, thus representing a coalition
strategy in $\tree$.

\medskip Fix $\game = (\arena,S)$ a parameterized safety game with
$\arena=\tuple{V, \Sigma, \Delta}$ and $v_0 \in V$ an initial
vertex. We assume for every $(v,v') \in V\times V$ such that
$\Delta(v,v') \ne \emptyset$, $\Delta(v,v')$ is given as a complete
DFA over $\Sigma$.  Those will be given as inputs to the algorithm.

Let $\tree = \langle {N,E, \ell_N,\ell_E} \rangle$ be the associated
unfolding tree with root $\node_0$.  For the rest of this section, we
fix an arbitrary ordering on the internal nodes of $\tree$ and on the
edges: $\nodeint = \{\node_1, \ldots, \node_m\}$ and
$E = \{e_1,\ldots, e_r\}$, with $|\nodeint| = m$ and $|E| = r$.

Assuming there are $t$ leaves --thus $t$ plays-- in $\tree$, for every
$1 \leq i\le t$, the $i$-th play is denoted
$\nodeij{i}{0} \ldots \nodeij{i}{z_i}$ with $\nodeij{i}{0} = \node_0$,
$\forall j<z_i, \ \nodeij{i}{j} \in \nodeint$ and
$\nodeij{i}{z_i} \in \nodeleaf$. Also, for $0 \le j < z_i$, we note
$e^i_j = (\nodeij{i}{j} ,\nodeij{i}{j+1})$.

The automaton for the winning coalition strategies in $\tree$ builds
on the finite automata that recognize the regular languages that label
edges of $\tree$. For each edge $e \in E$, let us write
$\B_{e} = (Q_{e}, \Sigma, \delta_{e}, q^0_{e}, F_{e})$ for the
complete DFA over $\Sigma$ such that $L(\B_e) = \ell_E(e)$. (Here
$Q_e$ is the set of states, $\delta_e$ the transition function,
$q^0_e \in Q_e$ the initial state and $F_e$ the set of accepting
states.) Note that some of the $\B_e$'s are identical since they
  correspond to the same original edge of $\game$.

We then define a deterministic safety automaton
$\B = (Q, \Sigma^m, \delta, q^0, F)$ 
that simulates all $\B_e$'s in parallel and
accepts $\omega$-words over alphabet $\Sigma^m$ 
if every prefix satisfies the following:
on every branch of the tree, if all corresponding $\B_e$'s
accept, then the leaf is labeled by a safe vertex. Formally,
$Q \subseteq Q_1\times\ldots \times Q_r$ is the set of states;
$q^0 = (q_1^0, \ldots, q_r^0)$ is the initial state; the transition
relation $\delta$ executes the $r$ automata $\B_e$'s componentwise: if
letter $u \in \Sigma^m$ is read, then make the $s$-th component mimick
$\B_{e_s}$ by reading the $l$-th letter of $u$, where $l$ is the index
(in the enumeration fixed above)
of the source node of $e_s$; and the accepting set $F$ is composed of
all states $q = (q_1, \ldots, q_r)$ that satisfy the following Boolean
formula:
\[
\varphi = \bigwedge\limits_{1 \le i \le t} \varphi_i
\quad\quad  \text{ where }
\varphi_i =  \left( \bigg[\bigwedge\limits_{0 \le j < z_i} q_{e^i_j} \in F_{e^i_j} \bigg] \Rightarrow \ell_N(\nodeij{i}{z_i}) \in S\right).
\]

Note that $\B$ is equipped with a safety acceptance
condition:\footnote{This is a slight abuse of language since $q^0$
  need not be in the safe set $F$.} an infinite run
$\run = q^0q^1q^2 \ldots $ of $\B$ is accepting if for every
$k \ge 1$, $q^k \in F$, and $L(\B)$ consists of all words $\word$
whose unique corresponding run is accepting.

Intuitively, $\varphi_i$ expresses that if for some number of agents
$k$, the languages along the $i$-th maximal path contain the
$k$-length prefixes of corresponding $\omega$-words (which means the
induced play is $k$-realizable), then it should lead to a safe
leaf
; and then $\varphi$ ensures that this should be true for all plays.
This is formalized in the next lemma.


\begin{restatable}{lemma}{bigautomaton}
\label{lem:automatonB}
Let  $\stratree: \nodeint \to \Sigma^\omega$ be a coalition strategy 
 in $\tree$. Then, $\stratree$ is winning if and only if
$(\stratree(\node_1),\stratree(\node_2),\ldots,\stratree(\node_m)) \in L(\B)$.
\end{restatable}

Notice that in the above statement, we slightly abuse notation:
$(\stratree(\node_1),\stratree(\node_2),\ldots,\stratree(\node_m))$
belongs to $(\Sigma^\omega)^m$, however it uniquely maps to a word in
$(\Sigma^m)^\omega$, that can thus be read in $\B$.
\begin{proof}
  Assume $\stratree: \nodeint \to \Sigma^\omega$ is a winning
  coalition strategy in $\tree$, and consider the corresponding word
  $\word =
  (\stratree(\node_1),\stratree(\node_2),\ldots,\stratree(\node_m))$.
  Let us show that $\word \in L(\B)$.  Consider the infinite run
  $\run = q^0q^1 \ldots $ of $\B$ on $\word$.  Fix a number of agents
  $k \in \pnats$.  Since $\stratree$ is winning, any $k$-length prefix
  of $\stratree$-induced play in $\koutree{\node_0,\stratree}$ is
  winning. Therefore for any $1 \le i \le t$ such that
  $\nodeij{i}{0} \ldots \nodeij{i}{z_i}$ is in
  $\koutree{\node_0,\stratree}$, the play satisfies for all
  $0 \le j < z_i$,
  $\wpref{\stratree(\nodeij{i}{j})}{k} \in \ell_E(e^i_j)$
  and furthermore,
  $\ell_N(\nodeij{i}{z_i}) \in S$; and hence $q^k \models \varphi_i$.
  Otherwise, if a length $k$-play is not induced by $\stratree$, then
  $\varphi_i$ is vacuously true for that $i\le t$.  We conclude
  $q^k \models \varphi$.  Since this is true for every $k\in \pnats$,
  $\word \in L(\B)$.

  Let now $\stratree$ be an arbitrary coalition strategy, and assume
  $\word =
  (\stratree(\node_1),\stratree(\node_2),\ldots,\stratree(\node_m))
  \in L(\B)$ with $\run = q^0q^1 \ldots $ the accepting run on
  $\word$.  Then for any number of agents $k \in \pnats$, $q^k \in F$,
  and hence $q^k \models \varphi$.  Therefore for all $1 \le i \le t$,
  $q^k \models \varphi_i$.  Fix any such $i$; let
  $\nodeij{i}{0} \ldots \nodeij{i}{z_i}$ be the $i$-th maximal path,
  and write $q^k = (q^k_1,\ldots, q_r^k)$.  Then for some
  $0 \le j < z_i$, the condition $q^k_{e^j_i} \in F_{e^j_i}$ implies
  $\wpref{\stratree(\nodeij{i}{j})}{k} \in \ell_E(e_j^i)$.
  In case the above is true for all $0 \le j < z_i$, we conclude
  $\nodeij{i}{0} \ldots \nodeij{i}{z_i} \in \koutree{\node_0,
    \stratree}$
  and $\varphi_i$ ensures that $\ell_N(\nodeij{i}{z_i}) \in S$.
  Otherwise,
  $\nodeij{i}{0} \ldots \nodeij{i}{z_i} \notin \koutree{\node_0,
    \stratree}$.  Finally $\varphi$ ensures all $k$-length prefixes of
  $\stratree$-induced
  plays in $\tree$ are winning.  Since this is true for any
  number of agents $k$, $\stratree$ is a winning coalition strategy in
  $\tree$.
\end{proof}

We now have all ingredients to solve the safety coalition problem, and
to state a complexity upper-bound. As mentioned earlier, we assume
that the arena is initially given with all associated complete DFAs
(used by all $\B_e$) in the input.
 
 \begin{restatable}{theorem}{expspace}
 \label{th:expspace}
  The safety coalition problem is in \EXPSPACE.
\end{restatable}

\begin{proof}
  Solving the safety coalition problem reduces to checking
  non-emptiness of the language recognized by the deterministic safety
  automaton $\B$. We adapt to our setting the standard algorithm which
  runs in non-deterministic logarithmic space, when $\B$ is given as
  an input.

  We write $N$ for the number of states
  of $\B$ and notice that $N$ is doubly exponential in $|V|$, the
  number of vertices of the initial arena $\arena$ (each state
  of $\B$ is an exponential-size vector of states of automata given in
  the input).  We do not build $\B$ a priori. Instead, we
  non-deterministically guess 
  a safe prefix of length at most $N$ (we only keep written two
    consecutive configurations and keep a counter to count up to
    $N$), and then a safe lasso on the last state of length at most
  $N$.

  Provided one can check `easily' whether a state of $\B$ 
  is safe, q the described procedure runs in non-deterministic
  exponential space, hence can be turned into a deterministic
  exponential space algorithm, by Savitch's theorem.

  It remains to explain how one checks that a given state in $\B$
  is safe. Formula $\varphi$ is a SAT formula exponential in the size 
 of $\arena$, which can therefore be solved in exponential
  space as well.

  Overall, we conclude that the safety coalition problem is in
  \EXPSPACE.
\end{proof}

\subsection{Synthesizing a winning coalition strategy}
\label{subsec:synthesis}

We assume all the notations of the two previous subsections, and we explain
how we build a winning coalition strategy.  From an
accepting word of the form $\mathbf{u} \cdot \mathbf{v}^\omega$ in
$\B$ (where $\mathbf{u} \in \big(\Sigma^m\big)^*$ and
$\mathbf{v} \in \big(\Sigma^m\big)^+$), one can synthesize a winning
strategy $\stratree$ in $\tree$ by: 
\[
  \stratree(\node_i) = \mathbf{u}_{i} \cdot \mathbf{v}^\omega_{i}
  \quad \text{for every $\node_i \in \nodeint$.}
\]
Then it is  easy to transfer to a winning coalition strategy
$\stratsys$ in $\game$ by defining
\[
\stratsys(h) = \stratree(\alpha(\zip(h))) \quad \text{for every history $h \in \Hist$,}
\]
that is, the $\omega$-word corresponding to the internal node
representing its virtual history. Recall that, following the proof of
Lemma~\ref{lem:str-eqv}, $\zip$ assigns to every history its virtual
history (by greedily removing all the loops) and $\alpha$ associates
to a virtual history its corresponding node in the tree $\tree$.

\begin{restatable}{proposition}{memory}
  If there is a winning coalition strategy for a game $\game =
  (\arena,S)$, then there is one which uses exponential memory, which
  can be computed in exponential space. Furthermore, winning might
  indeed require exponential memory.
\end{restatable}

\begin{proof}
  The tree unfolding can be seen as a memory structure for a winning
  strategy. Indeed, consider the memory set defined by $\nodeint$,
  starting from memory state $\node_0$. Define the application
  $\upd : \nodeint \times V \to \nodeint$ by $\upd(\node,v') = \node'$
  such that $v'\in S$ whenever
  \begin{itemize}
  \item either $\node' \in \nodeint$ is a child of $\node$ such that
    $\ell_N(\node') = v'$
  \item or $\node' \in \nodeint$ is an ancestor of
    $\node'' \in \nodeleaf$ such that
    $\ell_N(\node'') = \ell_N(\node') = v'$, and $\node''$ is a child
    of $\node$.
  \end{itemize}
  We also define the application
  $\next : \nodeint \times V \to \Sigma^\omega$ by
  $\next(\node,v) = \stratree(\node)$.

  Then, it is easy to see that winning strategy $\stratsys$ can be
  defined using memory $\nodeint$ and applications $\upd$ and $\next$.

Furthermore, though the $\omega$-words extracted fom $\B$ can
    be of doubly-exponential size, their computation and the overall
    procedure only requires exponential space.

\medskip For the lower bound, we show the following lemma.

\begin{lemma}
  There is a family of games $(\game_n)_n$ such that the size of
  $\game_n$ is polynomial in $n$ but winning coalition strategies
  require exponential memory.
\end{lemma}

\begin{proof}
  We again consider the game of Fig.~\ref{fig:worstcase-example},
  whose description can be made in polynomial time (since the $i$-th
  prime number uses only $\log(i)$ bits in its binary representation). We
  have already seen that its tree unfolding has exponential size. 
  We will argue why exponential memory is required, that is, one
  cannot do better than the tree memory structure.

  First notice that there is a winning coalition strategy: play
  $a^\omega$ at every vertex $B_i$, and $a^\omega$ (resp. $b^\omega$)
  at vertex $\vclause{i}$ if the history went through $v_i$
  (resp. $\bar{v}_i$). This strategy can be implemented using the
  memory given by the tree unfolding.

  Assume one can do better and have a memory structure of size
  strictly smaller than $2^n$. Then, arriving in vertex $\vclause{1}$,
  there are at least two different histories leading to the same
  memory state, hence the coalition strategy will select exactly the
  same $\omega$-words in all vertices $\vclause{1}$, $\vclause{2}$,
  ..., $\vclause{n}$. We realize that it cannot be winning since the
  two histories disagree at least on a predicate ``be a multiple of
  the $i$-th prime number''. Contradiction.
\end{proof}

\end{proof}

\subsection{Illustration of the construction}
\label{subsec:illustration}

We illustrate the construction on one example.

\input{figure}
\begin{example}
  \label{ex:B}
  Fig.~\ref{fig:big_automaton} represents part of the automaton $\B$
  corresponding to the tree $\tree$ in Fig.~\ref{fig:unfold2} (that
  is, the tree unfolding of the arena in
  Fig.~\ref{fig:example-arena}). The automata $\B_e$ for the languages
  labeling the edges of $\tree$ are depicted in
  Fig.~\ref{fig:automata_arena}.  Here notice that each state of $\B$
  has as many components as the number of edges leading to a safe node
  in $\tree$, we did not consider the edges leading to $\bot$. This is
  without loss of any generality: the language on any `unsafe' edge
  leading to $\bot$, in this example, are disjoint from the languages
  on the edges leading to its \emph{siblings} (other children of its
  parent node).  The first three positions in a state of $\B$,
  presented as a single cell in the picture, correspond to the
  outgoing edges of the \emph{root} $\node_0$ of $\tree$ (hence they
  follow the same component in $\Sigma^m$), and the other positions
  correspond to the other edges (in some chosen order). `$\times$' in
  a component of a state denotes the non-accepting sink state of the
  corresponding automaton (as mentioned in
  Fig.~\ref{fig:automata_arena}). Finally, here we have only shown the
  accepting states (marked in blue) and some of
  the 
  non-accepting states.
Indeed
  one can verify that the states which are colored in blue satisfy the formula $\varphi$; for
  instance, the  state ($p_1,p_1,\times,s_2,r_1,\times$)
  on the right 
  corresponds to the two maximal paths $v_0 v_0$ and $v_0 v_1 v_0$ 
  in $\tree$ (notice we used the node labels here),
  and all of them lead to safe nodes.  The infinite execution in blue
  (\emph{i.e.}, all the words in
  $(a,a,\Sigma,b)\cdot (b,a,\Sigma,\Sigma)\cdot
  (a,a,\Sigma,\Sigma)^\omega$) corresponds to the winning coalition
  strategies in the tree: for instance,
  $\stratree(\node_0) = aba^\omega$; 
  $\stratree(\node_1) = a^\omega$; 
  for any $a\in \Sigma$,
  $\stratree(\node_2) = a^\omega$; and
  $\stratree(\node'_1) = b^\omega$ is a winning coalition strategy
  (note here, for instance, that at node $\node_2$, any word from
  $\Sigma^\omega$ could be played).
\end{example}

%% file: figure.tex
\begin{figure}[h]
  \begin{minipage}{0.35\textwidth}
\begin{center}
  \begin{subfigure}[l]{\textwidth}
    \scalebox{0.9}{
      \begin{tikzpicture}[initial text =]
        \everymath{\scriptstyle}
        
        \node[state,minimum size=1mm,initial] at (0,0) (p0) {$p_0$};
        \node[state,minimum size=1mm,accepting] at (1.5,0) (p1) {$p_1$};
        \draw [-latex'] (p0) edge[ above] node{$b$} (p1);
        \draw [-latex'] (p1) edge[loop above] node{$a$} (p1);
        \draw [-latex'] (p0) edge[loop above] node{$a$} (p1);
      \end{tikzpicture}
    }
    \caption{Automaton for $a^*ba^*$.}
    \label{fig:aut2b}
  \end{subfigure} \\
  \begin{subfigure}[l]{\textwidth}
    \scalebox{0.9}{
      \begin{tikzpicture}[initial text =]
        \everymath{\scriptstyle}
        
        \node[state,minimum size=1mm,initial] at (0,0) (q0) {$q_0$};
        \node[state,minimum size=1mm,accepting] at (1.5,0) (q1) {$q_1$};
        \draw [-latex'] (q0) edge[ above] node{$a$} (q1);
      \end{tikzpicture}
    }
    \caption{Automaton for $a$.}
    \label{fig:aut1a}
  \end{subfigure} \\
  \begin{subfigure}[l]{\textwidth}
    \scalebox{0.9}{
      \begin{tikzpicture}[initial text =]
        \everymath{\scriptstyle}
        
        \node[state,minimum size=1mm,initial] at (0,0) (r0) {$r_0$};
        \node[state,minimum size=1mm,accepting] at (1.5,0) (r1) {$r_1$};
        \draw [-latex'] (r0) edge[ above] node{$\Sigma$} (r1);
        \draw [-latex'] (r1) edge[loop above] node{$\Sigma$} (r1);
      \end{tikzpicture}
    }
    \caption{Automaton for $\Sigma^+$.}
    \label{fig:aut2a}
  \end{subfigure} \\
  \begin{subfigure}[l]{\textwidth}
    \scalebox{0.9}{
      \begin{tikzpicture}[initial text =]
        \everymath{\scriptstyle}
        
        \node[state,minimum size=1mm,initial] at (0,0) (s0) {$s_0$};
        \node[state,minimum size=1mm] at (1,.7) (s1) {$s_1$};
        \node[state,minimum size=1mm, accepting] at (2,0) (s2) {$s_2$};
        \node[state,minimum size=1mm, accepting] at (1,-.7) (s3) {$s_3$};
        
        \draw [-latex'] (s0) edge[above] node[midway,sloped]{$a$} (s1); 
        \draw [-latex'] (s0) edge[below] node[midway,sloped]{$b$} (s3); 
        \draw [-latex'] (s1) edge[above] node[midway,sloped]{$a$} (s2); 
        \draw [-latex'] (s2) edge[loop below] node{$a$} (s2);
      \end{tikzpicture}
    }
    \caption{Automaton for $b \lor aa^+$.}
    \label{fig:aut3a}
  \end{subfigure}
\end{center}
\caption{Automata corresponding to the input languages of
  Fig.~\ref{fig:example-arena}.  The automata are not complete for
  sake of readability; all unspecified letters lead to a (sink)
  non-accepting state `$\times$'. \label{fig:automata_arena}}
\end{minipage}
\hfill
\begin{minipage}{0.6\textwidth}
{\centering
\scalebox{0.9}{
\begin{tikzpicture}[initial text=]
\node[inner sep=0pt,initial] at (0,0) (1) {
\begin{tabular}{c}
            $p_0$\\
            $p_0$\\
            $q_0$\\\hline 
            $s_0$\\\hline
            $r_0$\\\hline
            $s_0$\\
        \end{tabular}
        };
\node[inner sep=0pt] at (3,0) (2) {
\begin{tabular}{c}
            $p_0$\\
            $p_0$\\
            $q_1$\\\hline 
            $s_1$\\\hline
            $r_1$\\\hline
            $s_3$\\
        \end{tabular}
        };
\node[inner sep=0pt] at (6,0) (3) {
\begin{tabular}{c}
            $p_1$\\
            $p_1$\\
            $\times$\\\hline 
            $s_2$\\\hline
            $r_1$\\\hline
            $\times$\\
        \end{tabular}
        };
\node[inner sep=0pt] at (3,3) (4) {
\begin{tabular}{c}
            $p_1$\\
            $p_1$\\
            $\times$\\\hline 
            $s_1$\\\hline
            $r_1$\\\hline
            $s_3$\\
        \end{tabular}
        };
\node[inner sep=0pt] at (3,-3) (5) {
\begin{tabular}{c}
            $p_1$\\
            $p_1$\\
            $\times$\\\hline 
            $s_3$\\\hline
            $r_1$\\\hline
            $s_3$\\
        \end{tabular}
        };
\node[inner sep=0pt] at (6,-3) (6) {
\begin{tabular}{c}
            $p_1$\\
            $p_1$\\
            $\times$\\\hline 
            $\times$\\\hline
            $r_1$\\\hline
            $\times$\\
        \end{tabular}
        };
\node[inner sep=0pt] at (4,3) (dots1) {$\cdots$};
\node[inner sep=0pt] at (7,-3) (dots2) {$\cdots$};
        
\draw [rounded corners=.5em] (1.north west) rectangle (1.south east);
\draw [rounded corners=.5em,blue] (2.north west) rectangle (2.south east);
\draw [rounded corners=.5em,blue] (3.north west) rectangle (3.south east);
\draw [rounded corners=.5em] (4.north west) rectangle (4.south east);
\draw [rounded corners=.5em,blue] (5.north west) rectangle (5.south east);
\draw [rounded corners=.5em] (6.north west) rectangle (6.south east);

    \everymath{\scriptstyle}
\draw [-latex',blue] (3) edge[loop right] node{$
\begin{pmatrix}
           a \\
           a \\
           \Sigma \\
           \Sigma
         \end{pmatrix}$} (3);
\draw[-latex',blue] (1) to node[above] {$
\begin{pmatrix}
           a \\
           a \\
           \Sigma \\
           b 
         \end{pmatrix}$} (2);
\draw[-latex',blue] (2) to  node[above] {$
\begin{pmatrix}
           b \\
           a \\
           \Sigma \\
           \Sigma
         \end{pmatrix}$} (3);
\draw[-latex'] (1) to [bend left] node[above] {$
\begin{pmatrix}
           b \\
           a \\
           \Sigma \\
           b
         \end{pmatrix}$} (4);
\draw[-latex',blue] (1) to [bend right] node[below] {$
\begin{pmatrix}
           b \\
           b \\
           \Sigma \\
           b
         \end{pmatrix}$} (5);
\draw[-latex'] (5) to  node[above] {$
\begin{pmatrix}
           a \\
           \Sigma \\
           \Sigma \\
           \Sigma
         \end{pmatrix}$} (6);
\end{tikzpicture}
}}
\caption{Automaton $\B$ corresponding to the tree given in
  Fig.~\ref{fig:unfold2}.  Here we have only shown the accepting
  states (marked in blue) and some of the non-accepting states.
  Further explanations are given in Example~\ref{ex:B}.
  \label{fig:big_automaton}} 
\end{minipage}
\end{figure}

%% file: lower-bound.tex
We show the safety coalition problem is \PSPACE-hard by reduction from
\textsf{QBF-SAT}, which is known to be
\PSPACE-complete~\cite{StockmeyerM73}.  The construction is inspired
by the one in~\cite{BBM19}, where the first agent was playing against
the coalition of all other agents, with a reachability objective.

  \begin{proposition}
  \label{prop:PSP-hard}
  The safety coalition problem is \PSPACE-hard.
\end{proposition}

\begin{proof}
  Let $\varphi = \exists x_1 \forall x_2 \exists x_3 \ldots \forall x_{2r}
  \cdot \bigl(C_1 \wedge C_2 \wedge \ldots \wedge C_m \bigr) $ be a
  quantified Boolean formula in prenex normal form, where for every
  $1 \le h \le m$, $C_h = \ell_{h,1} \vee \ell_{h,2} \vee \ell_{h,3}$,
  and for every $1 \leq j \leq 3$,
  $\ell_{h,j} \in \{x_i, \neg x_i \mid 1 \le i \le 2r\}$ are the
  literals.

  In the reduction, we use sets of natural numbers (that represent the
  number of agents) corresponding to multiples of primes. Let thus
  $p_i$ be the $i$-th prime number and $M_i$ the set of all non-zero
  natural numbers that are multiples of $p_i$. For simplicity, we
  write $a^{M_i}$ to denote the set of words in $(a^{p_i})^+$, that is
  words from $a^+$ whose length is a multiple of $p_i$. It is
  well-known that the $i$-th prime number requires $O(\log(i))$ bits
  in its binary representation, hence the description of each of the
  above languages is polynomial in the size of $\varphi$.
  
\smallskip
  From $\varphi$, we construct an arena
  $\mathcal{A}_\varphi = \tuple{V,\Sigma,\Delta}$ as follows:
  \begin{itemize}
  \item
    $V = \{v_0,v_1,\ldots, v_{2r-1}, v_{2r}\} \cup
    \{\vvar{1},\vbvar{1},\ldots, \vvar{2r},\vbvar{2r}\} \cup
    \{\vclause{1},\vclause{2},\ldots,\vclause{m},\vclause{m{+}1}\}
    \cup \{\bot,\top\}$, where we identify some vertices:
    $v_{2r} = \vclause{1}$, and $\vclause{m{+}1} =\top$.
\item $\Sigma = \{a,b,c\} \cup \bigcup_{1\le i \le 2r} \{a_{i}\}$.
\item For every $0\le s \le r{-}1$, every $1 \le i \le 2r$ and every
  $1 \le h \le m$: 
  \begin{enumerate}
 { \item $\Delta(v_{2s},\vvar{2s{+}1}) = a^{M_{2s+1}}$ and
    $\Delta(v_{2s},\vbvar{2s{+}1}) = b^+ \setminus b^{M_{2s+1}}$}
    
  {\item $\Delta(v_{2s},\top) =  (a^+ \setminus a^{M_{2s+1}}) \cup b^{M_{2s+1}}$}
    
  \item  $\Delta(v_{2s+1},\vvar{2s{+}2}) = c^{M_{2s+2}}$ and
$\Delta(v_{2s+1},\vbvar{2s{+}2}) = c^+ \setminus c^{M_{2s+2}}$

  \item $\Delta(\vvar{i},v_{i}) = \Sigma^+$ and $\Delta(\vbvar{i},v_{i}) = \Sigma^+$

 {\item $\Delta(\vclause{h},\vclause{h{+}1}) = \bigcup_{1\le j \le 3} L_{h,j}$
 where $L_{h,j} =  a_i^{M_i}$ if    $\ell_{h,j} = x_i$; 
    $L_{h,j} = a_i^+ \setminus a_i^{M_i}$ if    $\ell_{h,j} = \neg x_i$.}
    
  \end{enumerate}
  To obtain a complete arena, all unspecified transitions lead to 
  vertex $\bot$.
\end{itemize}
On the arena $\arena_\varphi$, we consider the safety coalition game
$\game_\varphi = (\arena_\varphi,S)$ with
$S = V \setminus \{\bot\}$. The construction is illustrated on a
simple example with $3$ variables and $2$ clauses in Fig.~\ref{fig:PSP-hard}.

\begin{figure}[h]
\begin{center}
  \begin{tikzpicture}
    \everymath{\scriptstyle}

    \draw [state,inner sep=1pt,minimum size=5mm] (0,0) node [draw,fill=green!50] (v0) {$v_0$}; 
    \draw [state,inner sep=1pt,minimum size=5mm] (2.5,0) node [draw,fill=green!50] (v1) {$v_1$}; 
    \draw [state,inner sep=1pt,minimum size=5mm] (5,0) node [draw,fill=green!50] (v2) {$v_2$}; 
    \draw [state,inner sep=1pt,minimum size=5mm] (7.5,0) node [draw,fill=green!50] (v3) {$\vclause{1}$}; 

    \draw [state,inner sep=1pt,minimum size=5mm] (10,0) node [draw,fill=green!50] (vC2) {$\vclause{2}$}; 
    \draw [state,inner sep=1pt,minimum size=5mm] (12.5,0) node [draw,fill=green!50] (vC3) {$\top$}; 
    
    \draw [state,inner sep=1pt,minimum size=5mm] (2.5,1.5) node [draw,fill=green!50] (vx1) {$\vvar{1}$}; 
    \draw [state,inner sep=1pt,minimum size=5mm] (2.5,-1.5) node [draw,fill=green!50] (vxbar1) {$\vbvar{1}$};
        \draw [state,inner sep=1pt,minimum size=5mm] (2.5,2.5) node [draw,fill=green!50] (t1t) {$\top$}; 
    \draw [state,inner sep=1pt,minimum size=5mm] (2.5,-2.5) node [draw,fill=green!50] (t1b) {$\top$}; 
    \draw [state,inner sep=1pt,minimum size=5mm] (5,1.5) node [draw,fill=green!50] (vx2) {$\vvar{2}$}; 
    \draw [state,inner sep=1pt,minimum size=5mm] (5,-1.5) node [draw,fill=green!50] (vxbar2) {$\vbvar{2}$}; 
    \draw [state,inner sep=1pt,minimum size=5mm] (7.5,1.5) node [draw,fill=green!50] (vx3) {$\vvar{3}$}; 
    \draw [state,inner sep=1pt,minimum size=5mm] (7.5,-1.5) node [draw,fill=green!50] (vxbar3) {$\vbvar{3}$}; 
        \draw [state,inner sep=1pt,minimum size=5mm] (7.5,2.5) node [draw,fill=green!50] (t3t) {$\top$}; 
    \draw [state,inner sep=1pt,minimum size=5mm] (7.5,-2.5) node [draw,fill=green!50] (t3b) {$\top$}; 

\path[every node/.style={anchor=south,auto=false},-latex']
(t1t) edge[loop right] node [right] {$\Sigma^{+}$} (t1t)
(t1b) edge[loop right] node [right] {$\Sigma^{+}$}(t1b)
(t3t) edge[loop right] node [right] {$\Sigma^{+}$} (t3t)
(t3b) edge[loop right] node [right] {$\Sigma^{+}$}(t3b)
(vC3) edge[loop right] node [right] {$\Sigma^{+}$}(vC3)
;
    
    \draw [-latex'] (v0) -- (vx1) node [midway,above,sloped]
    {$a^{M_1}$}; \draw [-latex'] (v0) -- (vxbar1)  node
    [midway,below,sloped] {$b^+ \setminus b^{M_1}$}; 
    \draw [-latex'] (v1) -- (vx2)  node [midway,above,sloped] {$c^{M_2}$}; 
    \draw [-latex'] (v1) -- (vxbar2)  node
    [midway,below,sloped] {$c^+ \setminus c^{M_2}$}; 
    \draw [-latex'] (v2) -- (vx3)  node [midway,above,sloped] {$a^{M_3}$};
     \draw [-latex'] (v2) -- (vxbar3)  node
    [midway,below,sloped] {$b^+ \setminus b^{M_3}$}; 

    \draw [-latex'] (vx1) -- (v1) node [midway,left] {$\Sigma^+$}; \draw [-latex'] (vxbar1) -- (v1)  node [midway,left] {$\Sigma^+$}; 
    \draw [-latex'] (vx2) -- (v2)  node [midway,left] {$\Sigma^+$}; \draw [-latex'] (vxbar2) -- (v2)   node [midway,left] {$\Sigma^+$}; 
    \draw [-latex'] (vx3) -- (v3)  node [midway,left] {$\Sigma^+$}; \draw [-latex'] (vxbar3) -- (v3)   node [midway,left] {$\Sigma^+$};

    \draw [-latex'] (v0) ..controls +(60:1cm) and +(180:1cm).. (t1t) node [left,pos=.8] {$a^+ \setminus a^{M_1}$}; 
    \draw [-latex'] (v0) ..controls +(300:1cm) and +(180:1cm).. (t1b) node [left,pos=.8] {$b^{M_1}$}; 
    \draw [-latex'] (v2) ..controls +(60:1cm) and +(180:1cm).. (t3t) node [left,pos=.8] {$a^+ \setminus a^{M_3}$};
    \draw [-latex'] (v2)  ..controls +(300:1cm) and +(180:1cm).. (t3b) node [left,pos=.8] {$b^{M_3}$};

    \draw [-latex] (v3) .. controls +(50:1cm) and +(135:1cm) .. (vC2) node [midway,above] {$a_1^{M_1}$}; 
    \draw [-latex] (v3) -- (vC2) node [midway,above] {$ a_2^+ \setminus a_2^{M_2}$}; 
    \draw [-latex] (v3) .. controls +(-50:1cm) and +(-135:1cm)
    .. (vC2) node [midway,above] {$a_3^+ \setminus a_3^{M_3}$};

    \draw [-latex] (vC2) .. controls +(50:1cm) and +(135:1cm) .. (vC3) node [midway,above] {$a_1^{M_1}$}; 
    \draw [-latex] (vC2) -- (vC3) node [midway,above] {$a_2^+ \setminus a_2^{M_2}$}; 
    \draw [-latex] (vC2) .. controls +(-50:1cm) and +(-135:1cm)
    .. (vC3) node [midway,above] {$a_3^{M_3}$}; 
  \end{tikzpicture}
\end{center}
\caption{Parameterized arena for the formula
  $\varphi = \exists x_1 \forall x_2 \exists x_3 \cdot (x_1
  \vee \neg x_2 \vee \neg x_3) \wedge (x_1 \vee \neg x_2 \vee x_3)$.\\
  All unspecified transitions lead to the sink losing vertex
  $\bot$. Set $M_i$ denotes multiples of the $i$-th prime
  number. Vertex $\vvar{i}$ (resp. $\vbvar{i}$) represents setting
  variable $x_i$ to $\true$ (resp. $\false$). For any play reaching
  $\vclause{1}$, for every $i$, the number of agents is in $M_i$ iff
  the play went through $\vvar{i}$.}
  \label{fig:predicates}
  \label{fig:PSP-hard}
\end{figure}
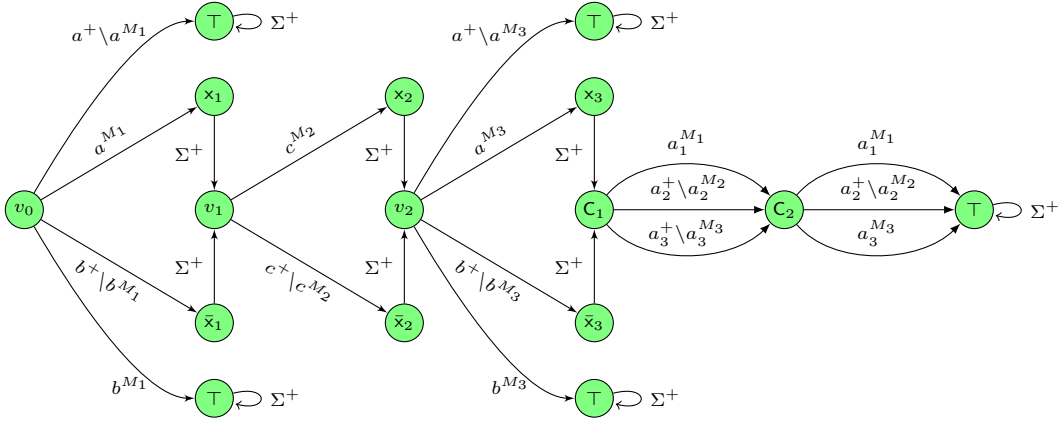

From $v_0$, a first phase up to $v_{2r} = \vclause{1}$ consists in
choosing a valuation for the variables. The coalition chooses the
truth values of existentially quantified variables $x_{2s+1}$ in
vertices $v_{2s}$: it plays $a^\omega$ for $\true$, and $b^\omega$ for
$\false$. In the first (resp. second) case, if the number of agents
involved in the coalition is (resp. is not) a multiple of $p_{2s+1}$,
then the game proceeds to the next variable choice, otherwise the safe
$\top$ state is reached (forever).

For universally quantified variables the coalition must play $c^\omega$ in
vertices $v_{2s{+}1}$, as any other choice would immediately lead to
the sink losing vertex $\bot$; the choice of the assignment then only
depends on whether the number of agents involved in the coalition is a
multiple of $p_{2s+2}$ (in which case variable $x_{2s+2}$ is assigned
$\true$) or not (in which case variable $x_{2s+2}$ is assigned
$\false$).

Hence, depending on the number of agents involved in the coalition,
either the play will proceed to state $v_{2r} = \vclause{1}$, in which
case the number of agents characterizes the valuation of the variables
(it is a multiple of $p_i$ if and only if variable $x_i$ is set to
$\true$); or it will have escaped to the safe state $\top$.

Note that in terms of information, the coalition learns progressively
assignments (thanks to the visit to either vertex $\vvar{i}$ or vertex
$\vbvar{i}$). Note also that the coalition can never learn assignments
of next variables in advance (it can only know whether it is a
multiple of previously seen prime numbers, hence of previously
quantified variables, not of variables quantified afterwards).

  \smallskip
From $\vclause{1}$, a second phase starts where one checks whether the
generated valuation makes all clauses in $\varphi$ true. If it is the
case, sequentially, the coalition chooses for every clause a literal
that makes the clause true. The arena forces these choices to be
consistent with the valuation generated in the first phase. For instance,
on the example of Fig.~\ref{fig:PSP-hard}, to set $x_1$ to $\true$
in the first phase, the coalition must play $a^\omega$, and only plays
with a number of agents in $M_1$ do not move to $\top$ and continue the
first phase from $\vvar{1}$. Then, in the second phase, for instance
for the first clause, one can choose literal $\ell_{1,1} = x_1$ by
playing $a_1^\omega$. The same language --${a_1}^{M_1}$-- labels the edge
from $\vclause{1}$ to $\vclause{2}$, so that the play proceeds to
$\vclause{2}$.
More generally, if $a_i^\omega$ leads from $\vclause{h}$ to
$\vclause{h{+}1}$ with number of agents in $M_i$, this means that
$\vvar{i}$ was visited, hence indeed $x_i$ was set to $\true$. On the
contrary, if $a_i^\omega$ leads from $\vclause{h}$ to
$\vclause{h{+}1}$ with number of agents not in $M_i$, this means that
$\vbvar{i}$ was visited, hence indeed $x_i$ was set to $\false$.

\medskip


To complete the proof sketch given above, we now formally
 prove that the reduction indeed ensures the following equivalence:
coalition has a winning strategy in the parameterized safety game
$\game_\varphi = (\arena_\varphi,S)$ if and only if
$\varphi$ is true.

\medskip For $\ell_{h,j}$ a literal of $\varphi$, we write
$\gamma(\ell_{h,j})$ for the constraint associated with $\ell_{h,j}$
in the second phase of the reduction. More precisely,

    \[
  \gamma(\ell_{h,j}) = \begin{cases}
     M_i & \textrm{if } \ell_{h,j}=x_i\\
     \neg M_i & \textrm{if } \ell_{h,j} = \neg x_i
  \end{cases}
\]

We rely on the following correspondence between histories from $v_0$
to $\vclause{1}$ and valuations over $\{x_1, \cdots, x_{2r}\}$: the
valuation $\iota_\pi$ associated with a history
$\pi = v_0 v'_1 v_1 v'_2 v_2 \cdots \vclause{1}$ is such that
$\iota_\pi(x_i) = 1$ if $v'_i = \vvar{i}$ and $\iota_\pi(x_i) = 0$ if
$v'_i = \vbvar{i}$. This correspondence defines a bijection, so that,
given a valuation $\iota$ over $\{x_1, \cdots, x_{2r}\}$, there is a
unique history $\pi_\iota = v_0 v'_1 v_1 v'_2 v_2 \cdots \vclause{1}$ with
$v'_i =\vvar{i}$ if $\iota(x_i) = 1$, and $v'_i = \vbvar{i}$

Moreover, after a finite history
$\pi = v_0 v'_1 v_1 v'_2 v_2 \cdots \vclause{1}$, the coalition can
deduce the number of agents is in the set
$K_\pi = \bigcap_{i\le2r}P_i$ where for each $i$, $P_i = M_i$ if
$v_i' = \vvar{i}$ and $P_i = \neg M_i$ if $v_i' = \vbvar{i}$.

There is thus also a bijection between valuations and possible
knowledge about the number of agents at $\vclause{1}$, so that we
abusively write $K_\iota$ when $\iota$ is a fixed valuation over
$\{x_1, \cdots, x_{2r}\}$. Note that for every $1 \le i \le 2r$,
$K_\iota \cap \{M_i,\neg M_i\}$ is either a subset of $M_i$ (when
$\iota(x_i)=1$) or a subset of $\neg M_i$ (when $\iota(x_i)=0$).
Also an extension of a history $\pi = v_0 v'_1 v_1 v'_2 v_2 \cdots \vclause{1}$
either leads to $\bot$, or continues in the main part of the game, in
which case it does not refine the knowledge set $K_\iota$ further.

\begin{lemma}
  \label{lem:qbfwin}
  Let $\iota$ be a valuation of $\{x_1, \cdots, x_{2r}\}$. Then
  $\iota \models C_1 \wedge \dots \wedge C_m$ if and only if coaltion has a
 winning strategy in $\game_{\varphi}$ starting from $\vclause{1}$,
if the number of agents $k$ is in $K_\iota$ .
\end{lemma}

\begin{proof}
  Assume $\iota \models C_1 \wedge \dots \wedge C_m$.  For every
  $1 \le h \le m$, there is $\ell_{h,j}$ (literal of $C_h$) such that
  $\iota(\ell_{h,j})=1$. We define the following strategy $\stratsys$
  for coalition: from vertex $C_h$, the strategy plays the
  word $\alpha_h^\omega$ defined by $\alpha_h = a_i$ if
  $\ell_{h,j} = x_i \text{ or } \neg{x_i}$.  By property on $K_\iota$,
  $K_\iota \subseteq \gamma(\ell_{h,j})$. Thus, $\sigma$ avoids
  $\bot$, and is therefore winning for every $k \in K_\iota$.

  Conversely, pick a winning strategy $\stratsys$ for coalition from
  $\vclause{1}$ against $K_\iota$. This strategy plays uniformly for
  all $k \in K_\iota$, say coalition chooses ${\alpha_h}^\omega$ from
  $\vclause{h}$. If for some $1\le i \le 2r$, $\alpha_h = a_i$, then
  either $x_i$ or $\neg x_i$ is a literal of $C_h$.  If
  $\Delta(\vclause{h},\vclause{h+1})$ contains ${a_i}^{M_i}$, then
  $x_i$ is the literal in $C_h$ and
  $K_\iota \subseteq \gamma(x_i)$; the construction guarantees that
  $\iota(x_i)=1$. A similar reasoning applies to the case where
  $\Delta(\vclause{h},\vclause{h+1})$ contains
  $a_i^+ \setminus {a_i}^{M_i}$ to show $\iota(x_i)=0$.  In all cases,
  $C_h$ is satisfied by $\iota$.
  \end{proof}

  Lemma~\ref{lem:qbfwin} formalizes the link between winning
  strategies for coalition in the second phase (that is, from
  $\vclause{1}$) and satisfying valuations. It remains to relate
  strategies for coaltion in the first phase and decision functions
  for the variable quantifications in $\varphi$.

  The two-player game underlying the quantifications of $\varphi$
  coincides with the coalition game played in the first phase. In
  particular, strategies coincide in the two games. Fix a coalition
  strategy $\stratsys$ in the quantification game or equivalently in
  the first phase. For every valuation $\iota$ which is generated by
  $\stratsys$ in the quantification game, there is an outcome $\pi$
  ending in $\vclause{1}$ such that $\iota = \iota_\pi$. Conversely,
  for every outcome $\pi$ ending in $\vclause{1}$, $\iota_\pi$ is
  generated by $\stratsys$ in the quantification game.
  
\end{proof}

%% file: conclusion.tex
\begin{figure}[h]
\centering
\vspace*{-.5cm}
\begin{tikzpicture}[shorten >=1pt,node distance=12mm and 3cm,on grid,auto,semithick]
    \everymath{\scriptstyle}
\node[state,inner sep=1pt,minimum size=5mm] (v_0) {$v_0$};
\node[state,inner sep=1pt,minimum size=5mm,double] (v_1) [right =2cm of v_0] {$v_1$};
\node[state,inner sep=1pt,minimum size=5mm] (v_2) [left = of v_0] {$v_2$};

 \path[every node/.style={sloped,anchor=south,auto=false},-latex']
(v_0) edge node [above] {$a^* b$} (v_1)
 (v_0) edge node [above] {$\Sigma^+ \setminus (a^* b + a^*ba^+)$} (v_2)
 (v_0) edge[loop above,-latex'] node [above] {$a^*b a^+$} (v_0)
  (v_1) edge[loop above,-latex'] node [above] {$\Sigma^+$} (v_1)
   (v_2) edge[loop above,-latex'] node [above] {$\Sigma^+$} (v_2);
\end{tikzpicture}
\caption{Example of a reachability coalition game: a winning coalition
  strategy is that Agent~$n$ plays $a$ for the first $n{-}1$ rounds,
  then $b$ for one round, and finally $a$ forever.}
\label{fig:example-arena-conclusion}
\end{figure}
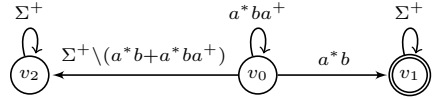

In this paper, we focused on and obtained results for the coalition
problem for safety objectives. The problem can obviously be defined
for other objectives. The finite tree unfolding technique will not be
correct in a general setting.
We illustrate this on the game arena in
Fig.~\ref{fig:example-arena-conclusion}.
In this example, the goal is to collectively reach the target
$v_1$. One can do so if, at $v_0$, \emph{the last} agent involved
plays a $b$ whereas all the others play $a$. On the other hand, at
$v_0$, it is safe if exactly one agent plays a $b$ and the other plays
an $a$. Coalition has a winning strategy: Agent~$n$ plays action $a$
for the first $n{-}1$ rounds, then plays $b$, and finally plays $a$
for the remaining steps. Doing so, each agent will in turn play action
$b$, and when the last agent does so, the play will reach
$v_1$. Notice that, for each agent (and hence for the coalition), the
strategy when going through $v_0$ differs at some step (at every step
from the coalition point-of-view), so no finite tree unfolding will be
correct.

As future work, we obviously would like to match lower and upper
bounds for the safety coalition problem, but more importantly we would
like to investigate more general objectives.